\documentclass[letterpaper, 10 pt, conference]{ieeeconf}  

\IEEEoverridecommandlockouts                              

\usepackage{xcolor}

\usepackage[hang]{subfigure}
\usepackage{epsfig,rotating,setspace,latexsym,amsmath,epsf,amssymb,amsfonts,bm,theorem,epstopdf}
\usepackage[normalem]{ulem}
\usepackage{cite,authblk}
\usepackage{algorithm}
\usepackage[noend]{algpseudocode}
\usepackage{setspace}
\usepackage{color}
\usepackage{graphicx}
\usepackage{mathrsfs}

\newtheorem{lemma}{Lemma}

\newtheorem{example}{Example}
\newtheorem{definition}{Definition}

\allowdisplaybreaks

\definecolor{green1}{rgb}{0.2,0.7,0.2}
\definecolor{brown}{rgb}{1,0.5,0.2}

\usepackage[colorinlistoftodos,textwidth=\marginparwidth]{todonotes}
\usepackage{comment}

\title{\LARGE \bf
Learning How to Strategically Disclose Information
}

\author{Raj Kiriti Velicheti,~ Melih Bastopcu,~ S. Rasoul Etesami,  ~Tamer~Ba{\c s}ar
\thanks{The authors are affiliated with the Coordinated Science Lab, University of Illinois Urbana-Champaign, Urbana, IL, USA 61801. Emails:
        {\tt\small \{rkv4,bastopcu,etesami1,basar1\}@illinois.edu}.}
\thanks{Research of the authors was supported in part by ARO MURI Grant AG285, and in part by AFOSR YIP Grant FA9550-23-1-0107 and by NSF CAREER Grant EPCN-1944403. 
}%
}

\begin{document}

\maketitle
\thispagestyle{empty}
\pagestyle{empty}

\begin{abstract}
Strategic information disclosure, in its simplest form, considers a game between an information provider (sender) who has access to some private information that an information receiver is interested in. While the receiver takes an action that affects the utilities of both players, the sender can design information (or modify beliefs) of the receiver through signal commitment, hence posing a Stackelberg game. However, obtaining a Stackelberg equilibrium for this game traditionally requires the sender to have access to the receiver's objective. In this work, we consider an online version of information design where a  sender interacts with a receiver of an unknown type who is adversarially chosen at each round. Restricting attention to Gaussian prior and quadratic costs for the sender and the receiver, we show that $\mathcal{O}(\sqrt{T})$ regret is achievable with full information feedback, where $T$ is the total number of interactions between the sender and the receiver. Further, we propose a novel parametrization that allows the sender to achieve $\mathcal{O}(\sqrt{T})$ regret for a general convex utility function. We then consider the Bayesian Persuasion problem with an additional cost term in the objective function, which penalizes signaling policies that are more informative and obtain $\mathcal{O}(\log(T))$ regret. Finally, we establish a sublinear regret bound for the partial information feedback setting and provide simulations to support our theoretical results.

\end{abstract}

\vspace{-0.25em}
\section{INTRODUCTION}
\vspace{-0.25em}
Asymmetric access to information is ubiquitous, especially in the growing world of online platforms. For example, consider the scenario where a platform wants to sell an advertisement slot (ad-slot) to a bidder. Although click-through rate is the most commonly used metric to measure the value of this slot, with the raising awareness of harms caused by misinformation, the bidders might also care about the information around which the advertisement is being placed. On the other hand, the seller might just want the slot to be sold for a higher price. While this mismatch in objective seems natural, as learning algorithms get better at understanding the content, the platform might have access to higher-quality information regarding the ad-slot. The platform, as the information provider, can reveal all this information to the bidder, who is the information receiver; however, due to the possible mismatch of their utilities, the platform could gain by revealing information strategically. Although the real situation is more complex with the presence of multiple bidders and dependent utilities, handling the simple case of a single bidder is instructive as an initial step.

Such strategic interactions occur in many other natural situations, such as disclosing the financial health of an institution, product ratings, and recommendation systems. While capturing such information exchange might be arbitrarily complex, a mathematical model for these interactions can be characterized, assuming that the players are rational and Bayesian. In the seminal work \cite{kamenica2011bayesian}, Kamenica and Gentzkow introduced a simple model to capture such interactions. Termed as Bayesian Persuasion, they consider a two-player game between a \textit{sender} who has access to some private information (state of nature) that is of interest to a \textit{receiver} whose action affects the utility of both players. While the sender cannot influence the receiver's actions with external transfers, they can shape the information that the receiver observes and thus influence their action. In contrast to the more classical framework of Strategic Information Transmission \cite{crawford1982strategic}, which considers a game with a similar setup and computes Nash equilibrium between the players, Bayesian Persuasion \cite{kamenica2011bayesian} is a Stackelberg game \cite{bacsar1998dynamic} with the sender moving first, and committing to a disclosure strategy. More generally, information design deals with games in which the designer (leader) can influence outcomes by cautiously shaping the beliefs of the players (followers).

Various extensions to the base model followed, including multiple senders \cite{gentzkow2016competition}, multiple receivers \cite{wang2013bayesian}, and multiple rounds of interaction between the sender and the receiver. We refer to \cite{kamenica2019bayesian} for a more comprehensive survey. Of particular interest to this work is costly signaling. More precisely, the original Bayesian Persuasion model assumes that the strategies are costless, and thus the sender can commit to a policy of arbitrary informativeness without facing any consequences other than their utility, which might not be true in most real-world examples. Information acquisition has innate costs, such as in the example described earlier, when gathering more information about the ad-slot might require training better learning models or paying third parties to gather such information. Hence, committing to an informative signal is costly. In an extension to their own work \cite{gentzkow2014costly}, the authors consider costly persuasion and solve for the optimal strategy of the sender. Note that the cost considered is on the signaling strategy and not on the message being communicated.

All the above works deal with the situation when the uncertain state belongs to a finite set, and the action space of the receiver is also finite. However, as stated in the leading example, we may have a continuous state and possibly multiple dimensions (e.g., attributes of a product). A different approach is required to solve such cases. Note that concavification becomes computationally intractable as the states of nature become large. While solving for general cost functions for infinite spaces is still an open problem, a general approach to deal with this is to assume quadratic utilities for the players and Gaussian uncertainty on the state, making the problem tractable even in the continuous setting \cite{tamura2018bayesian}. Recently, this approach has been extended to a dynamic game\cite{sayin2019hierarchical} and multiple senders \cite{velicheti2023value}, and a comparison with cheap talk has been analyzed in \cite{saritacs2016quadratic}. A central limitation of all the above models is that the sender is assumed to know the receiver's utility and thus can compute its best response and design its optimal signaling scheme accordingly. A recent work \cite{castiglioni2020online} addresses this limitation for the finite state problem by framing an online version of Bayesian Persuasion where the sender repeatedly interacts with a stream of adversarially selected receivers from a set of finite types. Unfortunately, finding an optimal sender's strategy, polynomial in the size of the state, is NP-hard, and thus the exponential complexity in state space is unavoidable.

Motivated by this complexity result, our goal is to investigate a case of continuous (and thus infinite) states and actions. First, we pose a \textit{costly} Bayesian Persuasion problem in the continuous state and action setting. Assuming a Gaussian prior and quadratic cost function along with costly signaling, we show that it is sufficient to consider linear plus noise policies for the sender when the signaling cost is inspired from the reduction in entropy of the estimation error. Hence, the problem, in spirit, boils down to an online convex program, which we show not only achieves a ($\mathcal{O}(\log T)$) sub-linear regret but also can be solved in polynomial time at each round. Further, we expand beyond quadratic costs and consider general convex costs for the sender but restrict our attention to linear policies. In this setting, we utilize a novel parametrization to obtain sub-linear regret. Finally, we provide simulations to support our results. We note that we let the receivers be chosen from a bounded set and hence move beyond countably finite types introduced in \cite{castiglioni2020online}. Further, due to the computational barrier of concavification beyond a few states, we require techniques that are different from prior works on online Bayesian Persuasion. Finally, we emphasize that we provide a black-box reduction of Bayesian Persuasion to an online convex optimization, and thus it can be easily used along with any efficient learning algorithm.

It is also worth noting that the learning techniques utilized in this paper are similar to the ones used in online matrix completion \cite{hazan2012near}. \cite{balcan2015commitment} considers a Stackelberg game where the leader plays a game against the follower of an unknown type who is adversarially chosen from a fixed set at each round. Although similar in setup, the sender in our game faces an infinite-dimensional optimization.

\subsubsection*{Notations}

We denote by $[m]$ the set $\{1,\dots,m\}$. $\mathbb{N}(0,\Sigma)$ denotes a Gaussian distribution with mean 0 and co-variance $\Sigma$, where $\Sigma$ is a positive semi-definite (psd) matrix. For a psd matrix $S$, we use $S^{\frac{1}{2}}$ to denote its unique psd square root. 
$Tr(.)$ denotes the trace of a matrix. 
For a given matrix $A$ and a vector $y$, $A^\top$ and $y^\top$ denote the transposes of that matrix and the vector, respectively.
The identity and zero matrices are denoted by $I$ and $O$, respectively. The determinant of a square matrix $A$ is denoted by $det(A)$. The Frobenius norm of a matrix Q is denoted by $\|Q\|_F$. $\Pi_{\mathscr{K}}(\cdot)$ denotes projection onto a set $\mathscr{K}$.
\vspace{-0.2em}
\section{System Model and Problem Formulation}\label{sec:System_model}
Consider a costly information disclosure (a.k.a. persuasion) game between a sender and a receiver, each with its own (possibly different) cost functions. The sender has access to the realization of an underlying state $\omega\in \mathbb{R}^d$ which is randomly drawn from a commonly known zero-mean Gaussian prior, i.e., $\omega\sim\mu_0(\omega)$, where $\mu_0(\omega)=\mathcal{N}(0,\Sigma_0)$. We consider multiple types of receivers, each having a quadratic cost $v_r^\theta(\omega,u)$, which depends on both the state $\omega$ and their own action $u\in\mathbb{R}^d$. Formally, let $v_r^\theta(\omega,u)=\|Q\omega+Ru\|^2$ with the parameters $\theta=\{Q\in\mathbb{R}^{r\times d},R\in\mathbb{R}^{r\times d}\}$ belonging to a bounded set $\max\{\|Q\|_F,\|R\|_F\}\leq\kappa$. We further assume that $R^\top R$ is invertible, i.e., $R$ is of full rank. On the other hand, the sender's cost $v_s$, in addition to being quadratic, is affected by the signaling policy that it selects. More precisely, $v_s(\eta)=(1-\beta)v(\omega, u)+\beta c(\eta)$, where $v(\omega,u)=\|Q_S\omega+R_Su\|^2$ and $\eta\in\Pi$ is a function from the policy space $\Pi$, which is a set of all Borel measurable functions from $\mathbb{R}^d$ to $\mathbb{R}^d$, $\eta$ is the sender's chosen signaling scheme, $y=\eta(\omega)$ is the signal realization, and $c(\eta)$ is the cost for choosing a signaling scheme, and $\beta$ is a fixed weight parameter that takes a value in $[0,1)$. In particular, if we introduce $h$, an  $\alpha$-strongly concave function of error covariance $\Sigma_e=\mathbb{E}[ee^\top]$ over its domain, with bounded gradients, where $e=\omega- \hat{\omega}$ and $\hat{\omega}=\mathbb{E}[\omega|y]$, then $c(\eta)=-h(\Sigma_e)$. Since $h$ is a concave function of $\Sigma_e$, higher information revelation is more costly for the sender. We assume that the gradient of $h(\Sigma_e)$ is bounded, i.e., we have $\|\nabla h(\Sigma_e)\|_F\leq G'.$ We provide the value of $G'$ as we consider specific $h(\Sigma_e)$ functions in the following sections. A single stage persuasion game proceeds as follows: 
\begin{itemize}
    \item The sender commits to a signaling policy $\eta(\cdot)$. 
    \item A state $\omega\sim\mathcal{N}(0,\Sigma_0)$ is sampled and the signal realization $y=\eta(\omega)$ is revealed to the receiver.
    \item The receiver takes a cost minimizing action. While the receiver incurs a cost of $v_r^\theta(\omega,u)=\|Q\omega+Ru\|^2$, the sender has a cost $(1-\beta) \|Q_S\omega+R_Su\|^2+\beta c(\eta)$, which is a convex combination of sender's actual cost function and the cost of signaling. 
\end{itemize}
\vspace{-0.2cm}
\subsection{Solving the Static Game}\label{subsec: solvestatic}
\vspace{-0.1cm}
A general procedure to solve the above static game for a fixed and known receiver type $\theta=\{Q,R\}$ is to pose it in the space of posterior covariances. This can be done by considering the receiver's best response. Since receiver's cost is convex, for a fixed sender's strategy, using the first-order condition, the best response of the receiver is given by $u=-\Lambda\hat{\omega}$, where $\Lambda=(R^\top R)^{-1}R^\top Q$, and $\hat{\omega}=\mathbb{E}[\omega|y]$ is the posterior after observing $y$. Considering the receiver's best response, in order to obtain the minimum cost, the sender needs to solve the following optimization problem: 
\begin{align}\label{eqn_obj_fun}
    \min_{\eta}(1-\beta)\mathbb{E}[\|Q_s\omega-R_s(R^\top R)^{-1}R^\top Q\hat{\omega}\|^2]+\beta c(\eta).\!\!
\end{align}
Due to the law of iterated expectations, we have $\mathbb{E}[\hat{\omega}^\top\Lambda^\top R_s^\top Q_s\omega]\!=\!\mathbb{E}[\hat{\omega}^\top\Lambda^\top R_s^\top Q_s\mathbb{E}[\omega|y]]\!=\!\mathbb{E}[\hat{\omega}^\top\Lambda^\top R_s^\top$ $Q_s\hat{\omega}]$. Finally, using the fact that $x^\top x \!= \!Tr(x^\top x\!) \!= \!Tr(x x^\top\!)$, we can rewrite the optimization problem in (\ref{eqn_obj_fun}) as 
\begin{align*}
\min_{\eta}(1\!-\!\beta)\mathbb{E}[\omega^\top Q_s^\top Q_s\omega]\!+\! (1\!-\!\beta)Tr(V\mathbb{E}[\hat{\omega}\hat{\omega}^\top])\!-\!\beta h(\Sigma_e),
\end{align*}
where $V=\Lambda^\top R_s^\top R_S\Lambda-\Lambda^\top R_S^\top Q_S-Q_S^\top R_S\Lambda$. 
Note that the first term, i.e.,  $\mathbb{E}[\omega^\top Q_s^\top Q_s\omega]$, is independent of sender's strategy. Since optimizing over the space of $\eta$ might be an intractable functional optimization problem, we use the following lemma to pose this problem as a Semi-definite Program (SDP).
\vspace{-0.6em}
\begin{lemma}\label{lem: etasig}\cite{9222530}
     Any signaling policy $\eta$ results in a posterior covariance $\Sigma$ satisfying $\Sigma_0\succeq\Sigma\succeq O$. Furthermore, any positive semi-definite matrix satisfying $\Sigma_0\succeq\Sigma\succeq O$ can be induced by a policy $\eta\in\Pi$. In particular, linear policies with noise are sufficient to induce any such matrix $\Sigma$ as posterior covariance. 
\end{lemma}
\vspace{-0.5em}
For completeness, we provide a simplified proof of Lemma~\ref{lem: etasig} in Appendix~A. As a result of Lemma~\ref{lem: etasig}, we can restrict our attention to policies of the form $\eta(\omega)=L\omega+w$, where $w\sim\mathcal{N}(0,\Sigma_w)$ is an independent Gaussian noise with a variance $\Sigma_w\succeq O$ chosen by the sender. 
Furthermore, since the prior is Gaussian, using linear signaling policies results in an error that has a zero-mean Gaussian distribution with a variance given by $\mathbb{E}[ee^\top]=\Sigma_e=\Sigma_0-\Sigma$, where $\Sigma=\mathbb{E}[\hat{\omega}\hat{\omega}^\top]$. Thus, finding the sender's optimal strategy is a convex optimization problem given by
\begin{align}\label{Eq: cvxstatic}
    \min_{\Sigma_0\succeq \Sigma_e\succeq O} J(\Sigma_e)=-(1-\beta)Tr(V\Sigma_e)-\beta h(\Sigma_e).
\end{align}
Further, since $h(\Sigma_e)$ is assumed to be strongly concave, and $Tr(V\Sigma_e)$ is a linear function of $\Sigma_e$, $J(\Sigma_e)$ is a strongly convex function with respect to $\Sigma_e$ for every $\beta>0$. Finally, choosing $\beta=0$ reduces the problem to the traditional Bayesian Persuasion. Since $J(\Sigma_e)$ is convex with respect to $\Sigma_e$ for all $\beta$ values (strongly convex when $\beta>0$), and the set $\Sigma_0\succeq \Sigma_e\succeq O$ is also convex, the optimization problem provided in (\ref{Eq: cvxstatic}) is a convex optimization problem. To illustrate this reformulation, we give the following example. 
\vspace{-1.6em}
\begin{example}  As described in the introduction, consider an information disclosure game between an online platform and a single advertiser. For a fixed incentive scheme, let the cost of both these players depend on the slot for which the information is being disclosed. For simplicity, let the slot be characterized by properties $\omega=[z,\theta_A]$ denoting the possible click-through rate ($z\in\mathbb{R}$), and the general quality of content it is placed with ($\theta_A\in\mathbb{R}$), respectively. Let the action $u$ of the advertiser be the price being quoted for the slot. While the platform might care more about the click-through rates of advertisements, the advertiser places different priorities on their reputation and click rates. This can be captured by assuming that the cost for the platform would be  $v_s(\omega,u)=\mathbb{E}[( z+0.5\theta_A-u)^2]$ while the advertiser has a cost given by $v_r(\omega,u)=\mathbb{E}[(\alpha_z z+\alpha_A\theta_A-u)^2]$ where $\alpha_z$ and $\alpha_A$ are some constants determined by advertiser types. The platform needs to figure out an optimal way in which the statistics or features of the slot need to be packaged to the advertiser so that its own cost is minimized. Utilizing the framework developed above in this situation results in
\begin{small}
\begin{align}\nonumber\\[-15pt]
 V = \begin{bmatrix}
\alpha_z^2-2\alpha_z & \alpha_z\alpha_A- 0.5\alpha_z-\alpha_A \\
\alpha_z\alpha_A- 0.5\alpha_z-\alpha_A & \alpha_A^2-\alpha_A 
\end{bmatrix}.   
\end{align}    
\end{small}
\begin{align*}
    \\[-50pt]
\end{align*}
\end{example}
In the next subsection, we take $c(\eta)$ to be the entropy function to characterize the signaling cost and provide a solution for this specific setting.
\subsection{Entropy Regularizer}\label{Sec: Entropyreg}
In the previous subsection, we formulated the optimization problem where the cost of sending information is represented by a general strongly concave function $h(\Sigma_e)$ (or equivalently, $-c(\eta)$ which is in terms of the chosen policy $\eta$). In this subsection, we consider $c(\eta)=-\mathbb{E}_\eta[-\log(f(e))]$ as the cost of selecting the policy $\eta$ where $f(\cdot)$ is the error distribution as a result of choosing policy $\eta$. In other words, $c(\eta)=-\mathbb{E}_\eta[-\log(f(e))]$ is the well-known differential entropy function, and thus it provides a measure of uncertainty in estimating the state. 

We now use the following lemma to prove that linear policies are indeed optimal in such a setting. Hence, the techniques from the previous section can be utilized to form an SDP to solve for the optimal sender's strategy.
\vspace{-0.7em}
\begin{lemma}\label{Lem: linopt}
   When $c(\eta)$ is chosen to be the differential entropy function, i.e., $c(\eta)=-\mathbb{E}_\eta[-\log(f(e))]$, then using linear plus noise policies, i.e., $\eta(\omega)=L\omega+w$ where $w\sim\mathcal{N}(0,\Sigma_w)$ are indeed optimal. 
\end{lemma}
\vspace{-0.5em}
\begin{proof}
    Consider a general policy $\zeta$ which results in an error covariance $\Sigma_e$, resulting in a total cost of $v_s(\zeta)=(1-\beta)[Tr(Q_s^\top Q_s\Sigma_0)+Tr(V(\Sigma_0-\Sigma_e))]-\beta h_\zeta(\Sigma_e)$. Let us consider a linear plus noise policy $\eta(\omega)$ which results in the same error covariance $\Sigma_e$ (such a policy is possible because of Lemma \ref{lem: etasig}), and the total cost obtained from the policy is given by $v_s(\eta)=(1-\beta)(Tr(Q_s^\top Q_s\Sigma_0)+Tr(V(\Sigma_0-\Sigma_e)))-\beta h_\eta(\Sigma_e)$. As mentioned earlier, using linear policies results in an error that has a Gaussian distribution. It is well known that for a given covariance $\Sigma_e$, Gaussian distribution maximizes the entropy \cite{cover1991information}. Thus, while the term $Tr(V(\Sigma_0-\Sigma_e))$ in $v_s(\eta)$ and $v_s(\zeta)$ is the same for both policies, the last term in the total cost function is minimized when we use linear policies, i.e., $ h_\eta(\Sigma_e)\geq h_\zeta(\Sigma_e)$. Thus, the total cost for policy $\eta$ is $v_s(\zeta)\geq v_s(\eta)$.
\end{proof}
Due to Lemma \ref{Lem: linopt}, when the cost of signaling is determined by the differential entropy function, i.e., $c(\eta)=-\mathbb{E}_\eta[-\log(f(e))]$, we can restrict our attention to linear plus noise signaling policies which are in the form of $\eta(\omega) = L \omega + w$. Using this class of policies results in a Gaussian distribution for the error, and thus the total cost is given by $J(\Sigma_e)=-(1-\beta)Tr(V\Sigma_e)-\beta\frac{1}{2}\log((2\pi e)^d det(\Sigma_e))$ by noting that $h(\Sigma_e)= \frac{1}{2}\log((2\pi e)^d det(\Sigma_e))$, which is the differential entropy for a Gaussian distribution with zero mean and covariance $\Sigma_e$ \cite{cover1991information}. Next, we prove that the differential entropy for a Gaussian distribution is strongly convex in the following sense.
\vspace{-0.4cm}
\begin{definition} A function $f:\mathcal{K}\rightarrow \mathbb{R}$ is $\alpha$-strongly convex with respect to the Frobenius norm if for any $X,Y\in \mathcal{K}$
\begin{align*}\\[-20pt]
    f(X)\geq f(Y)-Tr(\nabla f(X)(X-Y))+\frac{\alpha}{2}\|X-Y\|_F^2,
\end{align*}
\begin{align*}
    \\[-50pt]
\end{align*}
or equivalently, for any $t\in[0,1]$ we have
\begin{align*}\\[-20pt]
    \!f(tX\!+\!(1\!-\!t)Y)\!\leq \!tf(\!X\!)\!+\!(1\!-\!t)f(Y)\!-\!\frac{\alpha}{2}t(1\!-\!t)\|X\!-\!Y\|_F^2.
\end{align*}
\begin{align*}
    \\[-50pt]
\end{align*}
\end{definition}
We say $f$ is $\alpha$-strongly concave if $-f$ is $\alpha$-strongly convex.
\vspace{-0.9em}
\begin{lemma}
    When the cost of signaling is equal to the differential entropy function, i.e., $c(\eta)=-\mathbb{E}_\eta[-\log(f(e))]$, under linear policies, $J(\Sigma_e)$ is $\alpha$-strongly convex with respect to the Frobenius-norm over the domain $\!\mathscr{K}\!=\!\{\Sigma_e| \Sigma_0\!\succeq\!\Sigma_e\!\succeq\! O\!\}$.
\end{lemma}
\vspace{-0.5em}
\begin{proof}
    Since $\Sigma_e\!\!\in\!\!\mathscr{K}$, we have $(\Sigma_e)_{ii}\!\leq\!(\Sigma_0)_{ii}$, where $(\cdot)_{ij}$ denotes the element at the $i$th row and $j$th column of matrix $(\cdot)$. Let $\max_i (\Sigma_0)_{ii}=\gamma$, which implies that $(\Sigma_0)_{ii}\leq\gamma$ for all $i\in[d]$. Now, consider arbitrary matrices $X,Y\in\mathscr{K}$. We have 
    \begin{align*}\\[-20pt]
        \|X-Y\|_F&\leq\sqrt{d}\|X-Y\|_2\leq d^{\frac{3}{2}}\max_{i,j}\|X_{i,j}-Y_{i,j}\|.
    \end{align*}
    \begin{align*}
    \\[-50pt]
\end{align*}
Thus, there exists a pair $(i,j)$ with $i\in[d],  j\in[d]$, such that
\begin{align*}\\[-20pt]
    \|X_{i,j}-Y_{i,j}\|\geq\frac{\|X-Y\|_F}{4\gamma d^{\frac{3}{2}}}(X_{i,i}+Y_{i,i}+X_{j,j}+Y_{j,j}).
\end{align*}
    \begin{align*}
    \\[-47pt]
\end{align*}
By using \cite{christiano2014online}, this implies that the total variation distance between two Gaussian distributions with covariances $X$ and $Y$ is  $\Omega\big(\frac{\|X-Y\|_F}{4\gamma d^{\frac{3}{2}}}\big)$. Hence using \cite[Lemma 5.3 ]{moridomi2018online}, we obtain
 \begin{align*}\\[-20pt]
    -\log &\ det (tX+(1-t)Y)\leq-t \log det(X) \nonumber\\&-(1\!-\!t) \log det Y\!-\!\frac{t(1-t)}{2}\frac{1}{1152\gamma^2d^3\sqrt{e}} \|X-Y\|^2_F,
\end{align*}
    \begin{align*}
    \\[-47pt]
\end{align*}
which means that $-\log det$ is strongly convex. Thus, under linear policies with the differential entropy cost function,  $J(\Sigma_e)$ is an $\alpha$-strongly convex function with respect to the Frobenius norm, where $\alpha=\frac{1}{1152\gamma^2n^3\sqrt{e}}$.
\end{proof}
In the next section, we will utilize the strong convexity property of the total cost function, where we consider the online version of the problem, i.e., when the sender is uncertain about the receiver's type.

\vspace{-0.15cm}
\section{Online Linear Quadratic Persuasion}\label{Sect:Online_pers}
\vspace{-0.05cm}
Consider an online setting of the information disclosure problem in which the sender plays a repeated game with a stream of receivers of unknown types selected from the set $\max\{\|Q\|_F,\|R\|_F\}\leq\kappa$. The sender has access to a state that is randomly sampled from a zero mean Gaussian distribution at each round, i.e., $\omega_t\sim\mathcal{N}(0,\Sigma_0)$. Here, the sender commits to a signaling policy $\eta_t$ (which is possibly stochastic) without knowing the receiver's type. Then, after observing the sampled state, it sends a signal $y_t=\eta_t(\omega_t)$ to the receiver. The receiver then selects an action $u_t=\delta_t(y_t)$ that minimizes its cost. We are interested in two types of settings: (a) full feedback and (b) bandit feedback. In the first setting, after the sender selects its policy at time $t$, the receiver's type $\theta_t=\{Q^t,R^t\}$ is revealed to the sender. On the other hand, in the bandit feedback setting, only the scalar value of the sender's cost is revealed; hence, the sender does not observe the receiver's type. Further, we assume that the receiver's sequence of types is chosen adversarially before the game starts. In this setting, we are interested in developing an online algorithm that finds a signaling policy at round $t$ based on the information available until current time $t$. The performance measure for such an algorithm is given by a notion of regret defined as
\vspace{-0.5cm}
\begin{align}\label{eqn:regret}
R_T(\mathcal{A})\!=\!\!\sup_{\{\theta_1,\hdots,\theta_T\}}\!\!\Big\{\sum_{t=1}^T\mathbb{E}[v_s(\eta_t^\mathcal{A},\theta_t)]-\min_{\eta} \sum_{t=1}^Tv_s(\eta,\theta_t)\Big\},
\end{align}
where $\eta_t^\mathcal{A}=\mathcal{A}(\theta_1,\hdots,\theta_{t-1})$ is the output of algorithm $\mathcal{A}$ and $T$ denotes the total time. In the following subsection, we first consider the full information feedback setting, where the sender observes the receiver's type at the end of each round. Ideally, we would like to develop an algorithm whose regret grows sublinearly in $T$ while having a per-round running time that is polynomial in the problem parameters. We show that this is indeed possible in the following sections. 
\subsection{Full Information Feedback}
Developing an efficient algorithm for minimizing the regret function given in (\ref{eqn:regret}) directly is computationally intractable. Thus, we utilize techniques from Section \ref{subsec: solvestatic} to first reformulate the problem into an online convex optimization framework to provide a black-box reduction to standard online algorithms.
\begin{algorithm}[t]
\begin{small}
\caption{Online Gradient Descent for Bayesian Persuasion with Differential Entropy Cost}\label{alg:ogd_1}
\begin{algorithmic}[1]
\State{\textbf{Input:} Convex set $\epsilon,\Sigma_0,T,\Sigma_1\in\mathscr{K}$, where $\mathscr{K}=\{\Sigma|\Sigma_0\succeq\Sigma\succeq O\}$, with step sizes $\{\rho_t\}$.}
\State{\textbf{for $t=1$ to $T$ do}}
 \State{\quad\textbf{Choose}  $L_t,w_t$ such that $\Sigma_t\!=\!\Sigma_0L_t(L_t\Sigma_0L_t\!+\!\Sigma_{v,t})^{\dagger}L_t\Sigma_0$ \par and observe $\{Q_t,R_t\}$, where $\Sigma_t=\Sigma_0-\Sigma_{e,t}$.}
 \State{\quad\textbf{Update:}   $\Sigma'_{e,t+1}=\Sigma_{e,t}+\rho_t((1-\beta)V_t+\beta(\Sigma_{e,t}+\epsilon I)^{-1})$};
\State{\quad \textbf{Project:} $\Sigma_{e,t+1}=\Pi_{\mathscr{K}}(\Sigma'_{e,t+1})$;}
\State{\textbf{end}}
\end{algorithmic}
\end{small}
\end{algorithm}
 For simplicity of illustration, we restrict attention to the signaling cost inspired by entropy as stated in Section \ref{Sec: Entropyreg}. In particular, let $c(\eta)=-\ln(det(\Sigma_e+\epsilon I))$ for some $\epsilon>0$, which is a smoother version of differential entropy to account for the difference between entropy and differential entropy (as in \cite{christiano2014online}). Therefore, $J(\Sigma_e)=-(1-\beta)Tr(V\Sigma_e)-\beta \ln(det(\Sigma_e+\epsilon I))$. This is a strongly convex objective, and hence, online gradient descent \cite{zinkevich2003online} can provably achieve sub-linear regret. By utilizing Lemma~\ref{lem: etasig}, we can design linear plus noise policies to achieve $\Sigma_t$ at each round. This is summarized in Algorithm~\ref{alg:ogd_1}. We use the following lemma to state the regret guarantees for Algorithm~\ref{alg:ogd_1}.
 \begin{lemma}\label{lem: ogdconvergence}
 Under the assumptions made, when $c(\eta)$ is $\alpha-$strongly convex, the expected regret of Algorithm~\ref{alg:ogd_1} when $\beta>0$, with step sizes $\rho_t=\frac{1}{\alpha\beta t}$ is at most
 \begin{align*}\\[-20pt]
     R_T\leq\frac{G^2}{2\beta\alpha}(1+\log T).\\[-20pt]
 \end{align*}
 \end{lemma}
 \begin{proof} If we show that the gradients in each round are bounded, the convergence rate follows from \cite{hazan2007logarithmic}. First, we note that for every $\Sigma_e\in\mathscr{K}$, we have $Tr(\Sigma_e)\leq Tr(\Sigma_0):=D$. Now we bound the norm of gradients which is given by $\nabla J_t(\Sigma_e)=-(1-\beta)V_t-\beta(\Sigma_e+\epsilon I)^{-1}$. For that, we bound $\|V_t\|$ as follows:
 \begin{align*} \\[-20pt] \|V_t\|&\leq\|\Lambda_t^\top\Lambda_t-\Lambda_t^\top Q_S-Q_S^\top\Lambda_t\|_F\\&\leq\|\Lambda_t^\top\Lambda_t\|_F+\|\Lambda_t^\top Q_S\|_F+\|Q_S^\top\Lambda_t\|_F\\&\leq\|\Lambda_t\|_F^2+2\|\Lambda_t\|_F\|Q_S\|_F,\\[-20pt]
 \end{align*}
 where $\|\Lambda_t\|_F=\|R_S(R^\top R)^{-1}R^\top Q\|_F$, which is bounded above by $\|\Lambda_t\|_F\leq \|R_S\|_F\|(R^\top R)^{-1}\|_F\|Q\|_F$. Let $R^\top R=K$. Then, we have
\begin{align*}\\[-20pt]
\|K^{-1}\|_F\!=\!\sqrt{Tr(K^{-1}K^{-1})}=\sqrt{\sum_i\frac{1}{\lambda_i^2}}\!\leq\!\sum_i\frac{1}{\lambda_i}\leq\frac{d}{\lambda_{\min}},\\[-20pt]
\end{align*}
where $\lambda_{\min}$ is the minimum eigenvalue of $R^\top R$. Since $R$ is assumed to be full rank, $\lambda_{\min}>0$. Therefore, $\frac{d}{\lambda_{\min}}$ is finite. Let $K'=\frac{\|R_S\|_Fd\kappa}{\lambda_{\min}}$. Thus, $\|\nabla J(\Sigma_e)\|_F\leq (1-\beta)(K'^2+2K'\|Q_S\|_F)+\beta\frac{D}{\epsilon}=G<\infty$.
 \end{proof} 
Although Lemma \ref{lem: ogdconvergence} guarantees a $\mathcal{O}(\log{T})$ convergence, this heavily relies on the fact that $c(\eta)$ is an $\alpha$-strongly convex function with bounded gradients. Further, the convergence falls back to $\mathcal{O}(\sqrt{T})$ when $\beta=0$ as it reduces to an online linear program. In our setting, $\mathcal{O}(\sqrt{T})$ convergence rate is tight. As an example, consider the simple scalar instance with $\beta=0$, $R_S=R=-1$, $Q_S=1$, and $Q$ is chosen from the set $\{1,1+\sqrt{2}\}$ uniformly at random at each time instant. This implies that $V_t=\{-1,1\}$. Letting $\Sigma_0=\sigma_0$, since any online algorithm would result in $\mathbb{E}[Tr(V_t\Sigma_t)]=0$, the regret would remain $\Omega(\sqrt{T})$. Furthermore, in each round, Algorithm~\ref{alg:ogd_1} requires the computation of a gradient, which in the case of differential entropy cost function requires a matrix inversion, which is polynomial in the instance size. Thus, computational complexity is $\mathcal{O}(d^3)$ to calculate the gradient. Although polynomial in time, this is super-linear in the dimension of the matrix. Algorithms without projection steps, such as online conditional gradient \cite{hazan2012projection} can be utilized to retain computational efficiency while compromising on the regret bound. 
 
In summary, motivated by the traditional Bayesian Persuasion literature, in this section, we have focused on the problem where both the sender and the receiver have quadratic objective functions, but the sender does not know the receiver's type. However, the sender learns to play an efficient strategy online. Note that Algorithm~\ref{alg:ogd_1} requires the sender to have full knowledge of $V_t$ at each time step, which is possible in the full feedback setting. In the following subsection, we use a similar reduction to online optimization to give efficient rates for the bandit feedback setting.
\subsection{Bandit Feedback }
In this subsection, we consider the case where the sender cannot observe the receiver's type at the end of each round, and instead, the sender only observes their cost as a result of their signaling policy at time $t$. Since the sender does not observe the receiver's type, in contrast to the full information feedback described in the previous subsection, the sender cannot calculate the gradients of their utility directly. Thus, here, we provide a way to apply the online gradient descent algorithm provided in the previous subsection to the bandit feedback setting. For that, we apply the FKM algorithm, which is provided in Algorithm~\ref{alg:ogd_2}, (\cite[Algorithm~23]{hazan2016introduction}). For convenience, we assume that $\Sigma_0\succeq I$. If not, we can make a transformation of the set $\mathscr{K}$ such that the transformed set will include the identity matrix $I$. If we use the smoother version of the differential entropy, the cost function is upper bounded by $|J(\Sigma_e)|=|(1-\beta)Tr(V\Sigma_e)+\beta \ln(det(\Sigma_e+\epsilon I))|\leq B$ for some finite $B$ as both terms have finite values. 

\begin{algorithm}[t]
\begin{small}
\caption{FKM Algorithm for Bayesian Persuasion with Differential Entropy Cost with Bandit Feedback}\label{alg:ogd_2}
\begin{algorithmic}[1]
\State{\textbf{Input:} Decision set $\mathscr{K}=\{\Sigma|\Sigma_0\succeq\Sigma\succeq O\}$, $\Sigma_0$, $T$, set $ \Sigma_1 = 0$ with step sizes $\{ \delta, \rho\}.$}
\State{\textbf{for $t=1$ to $T$ do}}
 \State{\quad\textbf{Draw} $U_t\in \mathbb{S}_1$ uniformly at random, set $\Sigma_t' =\Sigma_t + \delta U_t$.}
 \State{\quad\textbf{Play}  $\Sigma_t'$, observe and incur loss $J_t(\Sigma_t')$. Let $G_t = \frac{d^2}{\delta}J_t(\Sigma_t')U_t$.}
\State{\quad \textbf{Update} $ \Sigma_{t+1} =\Pi_{\mathscr{K}_\delta}(\Sigma_t- \rho G_t)$}
\State{\textbf{end}}
\end{algorithmic}
\end{small}
\end{algorithm}
In the following lemma, we state the regret bound as a result of applying Algorithm~\ref{alg:ogd_2} in the bandit feedback setting. 
\vspace{-1.6em}
\begin{lemma}\label{lemma_bandit}
    Under the above assumptions and the bandit feedback setting, the expected regret of Algorithm~\ref{alg:ogd_2}, with step sizes $\rho = \frac{D}{d^2 T^{\frac{3}{4}}}$, $\delta = \frac{1}{T^{\frac{1}{4}}}$ is bounded by 
    \begin{align*}\\[-20pt]
        R_{T,bandit} \leq 9d^2 DGB T^{\frac{3}{4}}. \\[-20pt]
    \end{align*}
\end{lemma}
The proof of Lemma~\ref{lemma_bandit} follows directly from \cite[Theorem~6.6] {hazan2016introduction}, and noting that the sender's cost is bounded by $B<\infty$, $D$ and $G$ values are found in the proof of Lemma~\ref{lem: ogdconvergence}. We note that under the full-feedback setting, since the sender is able to find the gradient perfectly, it can apply the online gradient descent method efficiently. As a result, the sender obtains $O(\log(T))$ convergence rate. On the other hand, in the bandit limited case, since the sender cannot access the gradient at the end of each round, the convergence rate of the online algorithm reduces to $O(T^{\frac{3}{4}})$. In the following section, we consider the online Bayesian Persuasion problem under the more general class of cost functions for the sender.
 \section{Bayesian Persuasion with General Convex Costs}
 
In this section, we consider the original Bayesian Persuasion problem without signaling costs for ease of exposition. The solution approach illustrated in Section~\ref{subsec: solvestatic} relies heavily on forming an SDP, which can later be solved efficiently due to its convex nature. Although we assumed the sender's cost to be quadratic, it can be generalized to consider polynomial costs of the form
\begin{align}\label{eqn_sender_cost}\nonumber\\[-20pt]
v_s(\omega,u)=&a^\top\omega+b^\top u+c\omega^\top C^\top C \omega\nonumber\\&+du^\top D^\top D u+e\omega^\top E^\top E u,   
\end{align}
\begin{align*}
    \\[-45pt]
\end{align*}
while still retaining quadratic cost for the receiver. This implies that the receiver's best response is still $u=-\Lambda\hat{\omega}$. Then, by using the linearity of trace and the law of iterated expectation, the sender's optimization problem becomes 
\begin{align*}\\[-20pt]
    \!\min_\eta\mathbb{E}[v_s(\omega,u)]\!=&\min_\eta cTr(C^\top\!\! C\omega\omega^\top\!)\nonumber\!+\! dTr(\Lambda^\top\! D^\top\! D\!\Lambda \hat{\omega}\hat{\omega}^\top\!)\\&-eTr(E^\top E\Lambda\hat{\omega}\hat{\omega}^\top)=\min_{\Sigma_0\succeq S\succeq O}Tr(VS),\\[-20pt]
\end{align*}
where $V=d\Lambda^\top D^\top D\Lambda-eE^\top E\Lambda$. With the more generalized polynomial cost function considered in (\ref{eqn_sender_cost}), we see that the problem still falls into the domain of convex optimization and can be solved in an online setting as described in Section~\ref{Sect:Online_pers}. We note that one of the most commonly used bi-linear cost functions falls into this general class and thus can be solved by using the SDP approach.

However, the SDP approach heavily relies on the quadratic nature of the sender's cost, even for the more generalized bi-linear cost function mentioned above. We now introduce an alternative parameterization that is useful to solve problems with general convex costs (concave utilities) for the sender. We still restrict our attention to cases when the receiver's cost is such that its best response action $u$ is linear in the posterior mean of the state $\hat{\omega}$. In particular, this assumption is valid when the receiver's cost is quadratic, as in Section~\ref{sec:System_model}. Formally, this implies that the receiver's best response can be written as $u=-\Lambda\hat{\omega}$ for some $\Lambda$ which depends on the receiver's cost (type), $\hat{\omega}=\mathbb{E}[\omega|y]$. In general, for a given sender's cost function $v_s(\omega, u)$, the sender's optimal signaling policy can be arbitrary. In
this section, we restrict our attention to the case where the sender's policy is linear, that is $\eta(\omega)=L\omega$ for some $L\in\mathbb{R}^{d\times d}$.\footnote{Note that we restrict our attention to linear policies and do not consider noise here due to tractability in optimization. We leave further investigation of optimal strategies as future work.} This implies that $\hat{\omega}=\Sigma_0L(L\Sigma_0L^\top)^\dagger y$. Thus, the receiver's best response is given by $u=-\Lambda\Sigma_0L(L\Sigma_0L^\top)^\dagger L\omega$. Plugging this back into the sender's optimization problem gives $\!\min_{L\in\mathbb{R}^{d\times d}} v_s(\omega, \!-\Lambda\Sigma_0L(L\Sigma_0L^\top\!)^\dagger L\omega)$, which is not a convex function in $L$. To circumvent this issue, we let $\Sigma_0^\frac{1}{2}L(L\Sigma_0L^\top)^\dagger L\Sigma_0^\frac{1}{2}=W$ and then parameterize the sender's optimization problem as $\min_{W\in\mathbb{R}^{d\times d}}v_s(\omega,-\Lambda\Sigma_0^\frac{1}{2}W\Sigma_0^{-\frac{1}{2}}\omega)$, which is linear in $W$, and thus the objective is convex in $W$. The catch, however, is that $W$ belongs to a non-convex set as shown in following lemma.
\vspace{-0.6em}
\begin{lemma}\label{Lemma_5_1}
     Let $\mathscr{M}$ be the set of matrices given by $\mathscr{M}:=\{W|W = \Sigma_0^\frac{1}{2}L^\top(L\Sigma_0L^\top)^\dagger L\Sigma_0^\frac{1}{2}, L\in\mathbb{R}^{d\times d}\}$. Then, $\mathscr{M}$ is not a convex set. 
\end{lemma}
\vspace{-0.6em}
\begin{proof}
   Let $L' = \Sigma_0^\frac{1}{2}L$. Then, the set $\mathscr{M}$ can be written as $\mathscr{M}:=\{\Sigma|L'^\top(L'L'^\top)^\dagger L'\}$ for some $L'\in\mathbb{R}^{d\times d}$. This is of projection matrix form, and thus its eigenvalues can only be $\{0,1\}$. Thus, this set is non-convex. For example, let us consider $\Sigma_1 = U\Lambda_1U^\top \in \mathscr{M}$ and $\Sigma_2 = U\Lambda_2U^\top \in \mathscr{M}$, where $\Sigma_1 \neq \Sigma_2 $, and $\Lambda_1$, and $\Lambda_2$ are diagonal matrices with $(\Lambda_j)_{ii} = \{0,1\}$ for $j=1,2$. By the definition of convexity, any convex combination of these two matrices should also lie in the set $\mathscr{M}$, i.e., $\Sigma_3 = \alpha \Sigma_1 + (1-\alpha)\Sigma_2$ for any $\alpha \in (0,1).$ However, it is easy to see that $\Sigma_3 =U \Lambda_3 U^\top$ with $\Lambda_3 = \alpha \Lambda_1 + (1-\alpha)\Lambda_2$. Since there is at least one diagonal element of $\Lambda_3$ such that $(\Lambda_3)_{ii} \in (0,1)$, $\Sigma_3$ is not in the projection matrix form. Thus, the set $\mathscr{M}$ is not convex.
\end{proof}

In Lemma~\ref{Lemma_5_1}, we have shown that the set obtained by using the linear policies is not convex. However, we can still consider a convex relaxation of the set $\mathscr{M}$. We denote this set by $\mathscr{W}=\{W|I\succeq W\succeq O\}$. In particular, when the sender's utility function is linear, the optimal solution lies at the boundaries of this convex set, which are of projection matrix form (see Lemma~\ref{lem: projmat} in Appendix~B), and thus the optimum signaling policies are achievable using linear policies.

We now consider an online version of this problem with the sender's cost being a general convex cost function in a setting similar to Section~\ref{Sect:Online_pers}, where the sender plays a repeated game with a stream of receivers with varying cost functions. In each round, the sender faces a receiver with an unknown cost and commits to a signaling policy. The sender then collects a cost based on its committed signaling policy and learns the receiver's type at the end of each round (full feedback). As before, we consider bounded receiver types, or equivalently $\|\Lambda\|_F\leq\kappa'$. To indicate the dependence of the sender's cost on the receiver's action, we use a slight abuse of notation and write the sender's cost as $v_s(\Lambda)$ for a receiver of type $\Lambda$. 

For this problem, we utilize a variant of Follow the Perturbed Leader Algorithm \cite{kalai2005efficient} to obtain a sub-linear regret for the online information design with general convex costs for the sender. In particular, the algorithm requires solving a linear program at each iteration over the non-convex set $\mathscr{M}$. However, the key observation is that the optimal value for linear optimization over a convex set can be achieved at the extreme points of the convex set, which helps us relax our problem to optimization over the convex set $\mathscr{W}$. This relaxation makes the problem a convex optimization problem and thus can be solved using standard approaches. We summarize this approach in Algorithm~\ref{alg: FTPL}.
\begin{algorithm}[t]
\begin{small}
\caption{Follow the Perturbed Leader Algorithm for General Convex Cost Functions}\label{alg: FTPL}
\begin{algorithmic}[1]
\State{\textbf{Input:} $T,\Sigma_0, \mathscr{W}=\{W|I\succeq W\succeq O\}$, step size $\rho$}
\State{$W_1=\mathbb{E}_{N\sim\mathcal{D}}[\arg \max_{W\in\mathscr{W}}Tr(N W)]$}
\State{\textbf{for $t=1$ to $T$ do}}
 \State{\quad\textbf{Choose} $L_t$ such that $\Sigma_0^\frac{1}{2}W_t\Sigma_0^\frac{1}{2}=\Sigma_0L_t^\top(L_t\Sigma_0L_t)^{\dagger}L_t\Sigma_0$ and observe $\Lambda_t$}
 \State{\quad\textbf{Get utility} $v_s(\Lambda_t)$ and let $\nabla_t=\nabla v_s(\Lambda_t)$}
  \State{\quad\textbf{Update:}}  \begin{align}\label{eq: algupdate}\nonumber\\[-25pt]
\!W_{t+1}\!=\!\mathbb{E}_{N\sim\mathcal{D}}[\arg\min_{W\in\mathscr{W}}\rho\sum_{s=1}^t Tr(\nabla_sW)+Tr(NW)\}]
\end{align}
\begin{align*}
    \\[-45pt]
\end{align*}
\State{\textbf{end}}
\end{algorithmic}
\end{small}
\end{algorithm}

In contrast to utilizing strongly convex regularizing functions, Algorithm~\ref{alg: FTPL} utilizes perturbed cost functions to add randomization to decision-making at each time instant. In particular, since $N$ is symmetric, as illustrated in (\ref{eq: algupdate}), upper triangle elements of $N$ are sampled from a distribution $\mathcal{D}$. Further, assuming $\mathcal{D}$ to be uniform over $[0,1]^{d(d+1)/2}$ the distribution $\mathcal{D}$ turns out to be $(\sigma,L)$ stable with $\sigma\leq d$ and $L\leq 1$ with respect to Frobenius norm, where $\mathbb{E}_{N\sim\mathcal{D}}[\|N\|_F]=\sigma$ and $\int_N|\mathcal{D}(N)-\mathcal{D}(N-U)|d N\leq L\|U\|_F$ for all matrices $U\in \mathscr{W}$ and $\mathcal{D}(N)$ denotes the value of probability density function $\mathcal{D}$ over $N$. Consistent with the assumption of an oblivious adversary selecting the receiver at each instant, it is assumed that the adversary does not have access to the realization of this random matrix. Hence, Algorithm \ref{alg: FTPL} attains the regret bound $\mathcal{O}(2Gd\sqrt{T})$, where $\|\nabla v_s\|_F\leq G$ for all $t\in T$ \cite{hazan2016introduction}. Although sub-linear in $T$, we note that the bound has a linear dependence on the problem dimension $d$, which can be optimized by choosing a better perturbation distribution $\mathcal{D}$. Furthermore, the algorithm requires solving a linear program at every instant. One possible way to reduce the computational complexity of the algorithm further is to utilize Follow the Lazy Leader algorithm  \cite{kalai2005efficient}, which might decrease the number of switches and thus can avoid the optimal decision computation repeatedly. 
\vspace{-0.1cm}
\subsection*{Solving the Offline Problem for General Convex Costs}
Although we achieve a sub-linear regret for the online version, this heavily relies on the fact that when the cost function $v_s$ is a convex function of $W$, and thus we have $\sum_{t=1}^Tv_t(W_t)-\sum_{t=1}^Tv_t(W^*)\leq\sum_{t=1}^T\nabla_{t}^\top(W_t-W^*)$, where $v_t=v_s(\Lambda_t)$, which makes the problem amenable for solution with a linear optimization at every time step. The achievability of optimal decision at every time step, even with relaxed convex constraints on the boundaries, helps us solve the problem. However, a natural question arises as to what policy we are comparing against in hindsight. Since the set we considered originally was $\mathscr{M}$, we would be competing against the best linear strategy for the sender in hindsight, knowing all the costs $v_t(\cdot)$ over time. Applying similar tools may not be useful to solve this problem in an offline setting. However, we utilize the following approach to give an approximation to the offline problem. First, we solve the relaxed convex optimization problem and let the solution of the relaxed offline optimization problem be given by $W^*$, i.e.,  $W^* = \text{arg}\!\min_{W\in\mathscr{W}}\mathbb{E}[v(W)]$. In the following lemma, we show that by using randomized linear policies (i.e., using a probabilistic mixture of fixed linear policies), the sender can achieve $W^*$ on average.  

\vspace{-0.5em}
\begin{lemma}
    Any $W\in\mathscr{W}$ can be achieved by using randomized linear policies. 
\end{lemma}
\vspace{-0.5em}
\begin{proof}
    Consider an arbitrary $W$ such that $I\succeq W\succeq O$. Since it is a positive semi-definite matrix, we can write this as $W=U\Delta U^\top$. Further, since $I\succeq W$, $rank(W)\leq d$, where $d$ is the dimension of the state. Further, notice that if $rank(W)=d$, then $W=I$, which can be achieved using the linear policy of $L=I$. If $rank(W)<d$, then consider the vector $\Delta\in[0,1]^d$, which is a diagonal matrix, and hence can be viewed as a vector in unit polytope which lies inside the convex hull of vectors $\Delta_i=\begin{bmatrix} 0,\hdots,1,\hdots,0\end{bmatrix}$ with 1 at $i$th position, and thus by Carathedory's theorem \cite{caratheodory1911variabilitatsbereich}, we can write $\Delta$ as a convex combination of $\Delta_i$. Finally, utilizing matrices $W_i=U\Delta_iU^\top$ and finding $L_i$'s corresponding to such $W_i$'s completes the proof.
\end{proof}

If $v$ is linear, such policies would imply that we can achieve $v(W^*)$ in expectation. But if we assume that $\nabla v\leq G$, then using these randomized policies gives us an approximation of the order $Gd$, which can be arbitrarily bad compared to the relaxed objective as dimensions increase. Another possible solution would be to use iterative plane cutting methods to directly search for an optimal solution in the space of projection matrices. We leave this for a future investigation.

As an example, we can consider $v_s(\omega,u)$ to be quadratic as in the previous sections. Specifically, let $v_s(\omega,u)=\|Q_S\omega+R_Su\|^2$. Note that for such quadratic costs, we can assume linear policies for the sender without loss of generality since the optimal hindsight policy is also purely linear (without additive noise). Hence, considering the parametrization in this section, the objective can be written as $\|(Q_s-R_S\Lambda W)\omega\|^2=\omega^\top(Q_S-R_S\Lambda W)^\top(Q_S-R_S\Lambda W)\omega=Tr((Q_S-R_S\Lambda W)^\top(Q_S-R_S\Lambda W)\omega\omega^\top)$. Therefore $\mathbb{E}[v_s(\omega,u)]=\|(Q_S-R_S\Lambda W)\Sigma_0^\frac{1}{2}\|_F^2$. Further, if $Q$ is assumed to be of full rank, then this is a strongly convex function. However, the set over which $W$ needs to be optimized in the offline problem is still non-convex, as described in this section. 

\vspace{-0.20cm} \section{Numerical Results}\label{sect:num_res}
\vspace{-0.15cm}
 In this section, we provide numerical results to verify the theoretical analysis of our work. In the first numerical example, we consider a sender and a stream of receivers and a state of the world given by $\omega\!\!=\!\!\begin{bmatrix}z \!& \!\!\theta_A\!\! &\!\!\theta_B\end{bmatrix}^\top$\!\!. In this example, we consider quadratic cost functions for receivers as
\begin{align}\nonumber\\[-1.8em]
 v_r(\omega,u) =\mathbb{E}[||z+\alpha_A\theta_A+\alpha_B\theta_B-u||^2],\label{num_res_cost}
\end{align}
\begin{align*}\\[-4.5em]
\end{align*}
where the receivers' quadratic cost function depends on their types determined by parameters $(\alpha_A, \alpha_B)$. We choose $\alpha_A$ uniformly randomly in the interval $[0.1, 4.1]$, and $\alpha_B$ uniformly randomly in the interval $[0.2, 4.2]$. The sender has a fixed cost function, which is given by $v_s(\omega,u) = (1-\beta) \mathbb{E}[||z+ \theta_A-u||^2]+\beta c(\eta),$ where we set the cost of signaling to be the differential entropy function $c(\eta) =-\mathbb{E}_\eta[-\log(f(e))]$. If we substitute the receiver's best response into the sender's cost function, we get  
\begin{align}\nonumber\\[-1.8em]
     J_s(\Sigma) = (1-\beta) Tr(V \Sigma) -\beta \log (\text{det}(\Sigma_0-\Sigma)),\label{sender's_obj}
\end{align}
\begin{align*}\\[-4.5em]
\end{align*}
where for a given receiver's type $(\alpha_A, \alpha_B)$, $V$ is equal to
\begin{small}
\begin{align}\label{V_mat}
    V = \begin{bmatrix}
-1& -1 & 0  \\
-1& (\alpha_A-1)^2-1 & \alpha_B(\alpha_A-1)\\
0 & \alpha_B(\alpha_A-1) & \alpha_B^2
\end{bmatrix},
\end{align}
\end{small}and $\Sigma=\Sigma_0-\Sigma_e$. In this example, we take $\beta=0.5$, and vary the total time duration $T\in[2,100].$ We apply Algorithm~\ref{alg:ogd_1} to find the sender's average regret. We plot the sender's average regret in Fig.~\ref{Fig:sims_1_2}(a). We see that for the considered sender's and receiver's cost functions, the sender's average regret reduces to 0 with a rate of $\mathcal{O}\left(\frac{\log(T)}{T}\right)$.  
\begin{figure}[t]
 	\begin{center}
 	\subfigure[]{%
 	\includegraphics[width=0.47\linewidth]{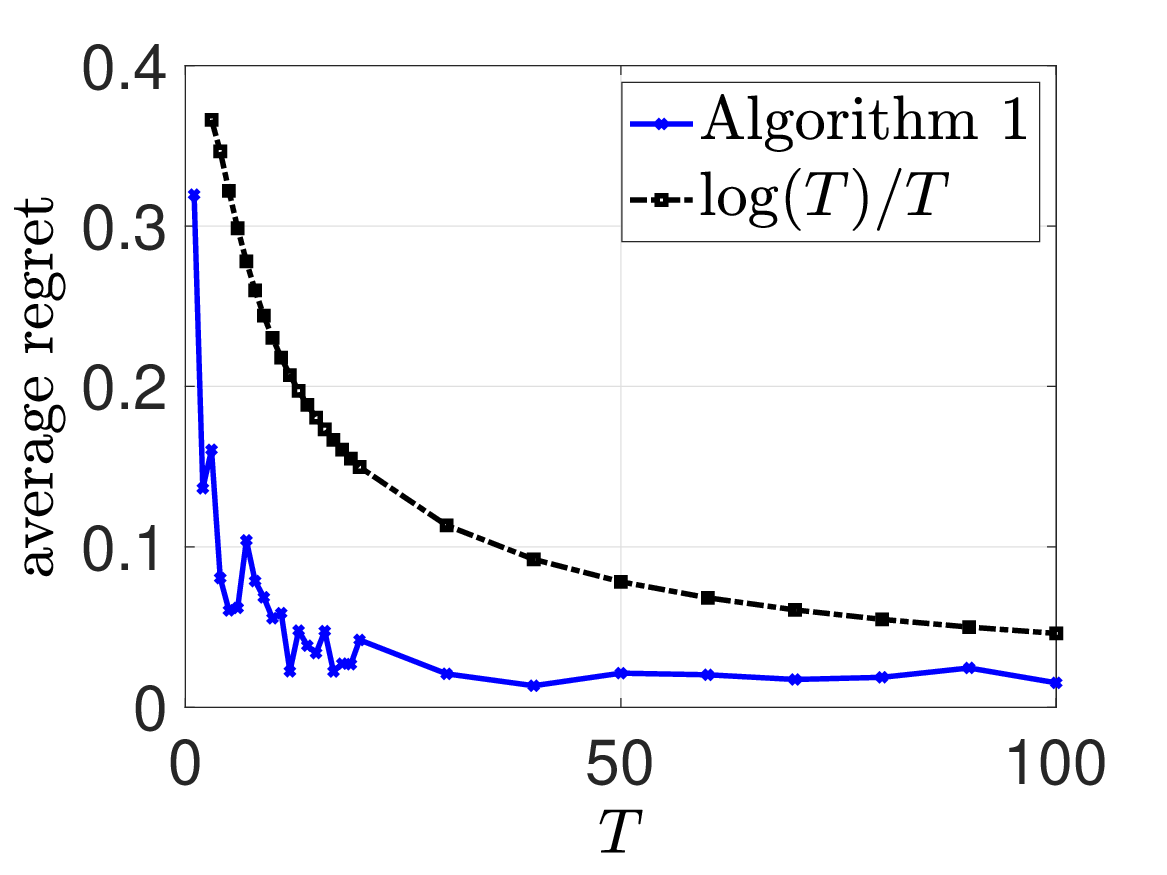}}
 	\subfigure[]{%
 	\includegraphics[width=0.47\linewidth]{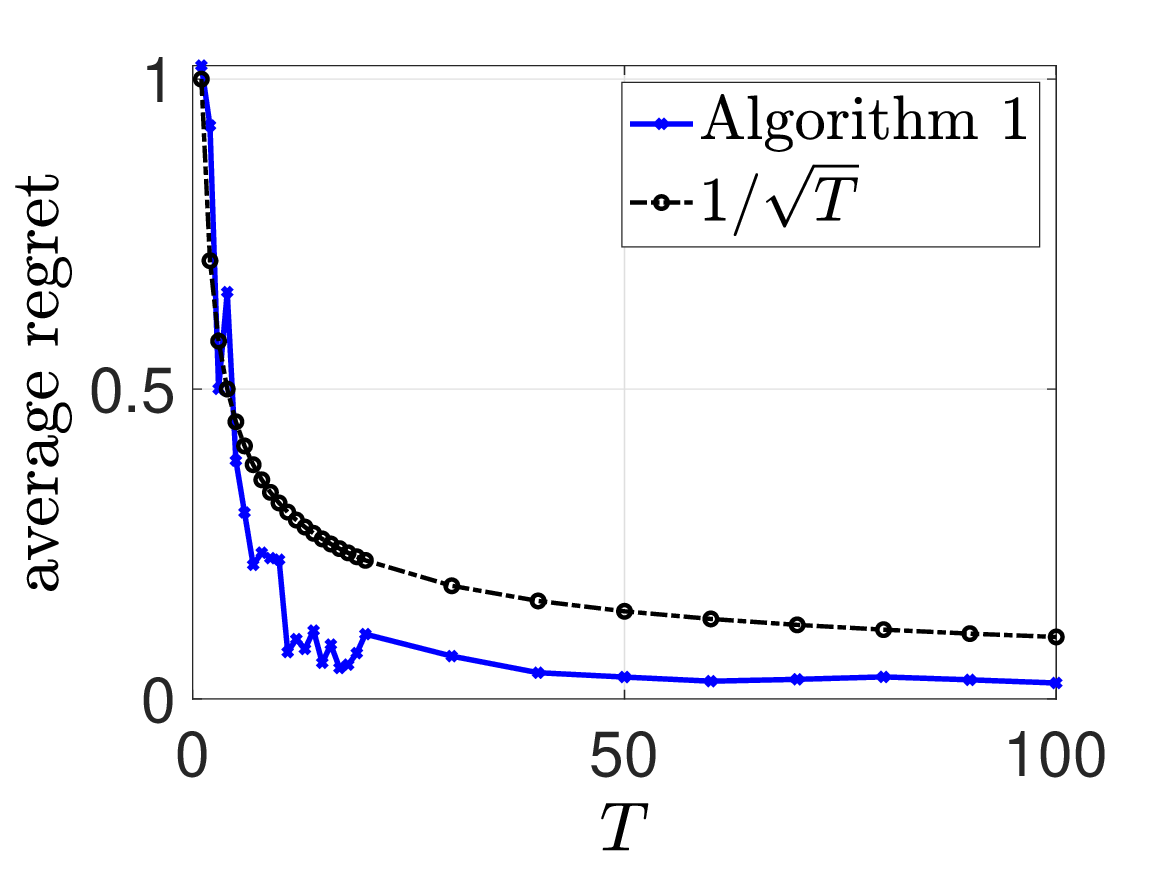}}
 	\end{center}
 	\vspace{-0.4cm}
 	\caption{The sender's average regret as a result of Algorithm~\ref{alg:ogd_1} as $T \in [1,100]$ (a) when $\beta = 0.5$ and (b) when $\beta = 0$.} 
 	\label{Fig:sims_1_2}
 	\vspace{-0.3cm}
 \end{figure}
Next, we consider the same structure as in the first numerical result, but this time, we take $\beta=0$. In this case, there is no signaling cost for the sender, and the sender's optimization problem turns out to be the traditional Bayesian Persuasion, which is given by $v_s(\Sigma) =  Tr(V \Sigma)$, where $V$ matrix is given in (\ref{V_mat}). We plot the sender's average regret in Fig.~\ref{Fig:sims_1_2}(b). As the optimization problem is no longer strongly convex, we see that the convergence rate for the online traditional Bayesian persuasion problem reduces to $\mathcal{O}\big(\!\frac{1}{\sqrt{T}\!}\big)$.
  
 \vspace{-0.5em}
 \section{Conclusion and Discussion}
\vspace{-0.3em}
In this work, we have considered an online version of information design where the sender does not know the receiver's cost function, and hence, needs to learn to play online. In the case of continuous state and action spaces, this amounts to learning the optimal directions of information revelation in the Euclidean space. First, restricting attention to quadratic cost functions, we have shown that in contrast to the discrete state case \cite{castiglioni2020online}, a sublinear regret is achievable with per round running time, which is polynomial in problem dimension. Furthermore, we have formulated a costly version of Bayesian persuasion where the sender's cost increases for choosing a more informative signal, and we have shown that linear policies are optimal when the sender has entropy costs. We have shown, both theoretically and empirically, that this regularized cost function gives a $\mathcal{O}(\log T)$ convergence. Finally, we have generalized our approach to general convex functions for the sender while restricting attention to linear strategies for the sender. As the approach used in earlier work is not adaptable to such general costs, we have developed a novel parameterization and proved a sublinear regret.

In this work, we have restricted our attention to sampling state independently at every time instant. In more general scenarios, such as in recommendation systems, multiple states might evolve in a Markovian form, and the techniques developed in this work would act as an initial step toward analyzing such settings. Furthermore, extensions to priors beyond Gaussian would be another viable research direction.

\vspace{-0.65em}
\bibliographystyle{unsrt}
\bibliography{IEEEabrv,sample-bibliography}

\begin{thebibliography}{10}

\bibitem{kamenica2011bayesian}
E.~Kamenica and M.~Gentzkow.
\newblock Bayesian persuasion.
\newblock {\em American Economic Review}, 101(6):2590--2615, 2011.

\bibitem{crawford1982strategic}
V.~P. Crawford and J.~Sobel.
\newblock Strategic information transmission.
\newblock {\em Econometrica: Journal of the Econometric Society}, pages 1431--1451, 1982.

\bibitem{bacsar1998dynamic}
T.~Ba{\c{s}}ar and G.~J. Olsder.
\newblock {\em Dynamic {N}oncooperative {G}ame {T}heory}.
\newblock SIAM, 1998.

\bibitem{gentzkow2016competition}
M.~Gentzkow and E.~Kamenica.
\newblock Competition in persuasion.
\newblock {\em The Review of Economic Studies}, 84(1):300--322, 2016.

\bibitem{wang2013bayesian}
Y.~Wang.
\newblock Bayesian persuasion with multiple receivers.
\newblock {\em Available at SSRN 2625399}, 2013.

\bibitem{kamenica2019bayesian}
E.~Kamenica.
\newblock Bayesian persuasion and information design.
\newblock {\em Annual Review of Economics}, 11:249--272, 2019.

\bibitem{gentzkow2014costly}
M.~Gentzkow and E.~Kamenica.
\newblock Costly persuasion.
\newblock {\em American Economic Review}, 104(5):457--462, 2014.

\bibitem{tamura2018bayesian}
W.~Tamura.
\newblock Bayesian persuasion with quadratic preferences.
\newblock {\em Available at SSRN 1987877}, 2018.

\bibitem{sayin2019hierarchical}
M.~O. Sayin, E.~Akyol, and T.~Ba{\c{s}}ar.
\newblock Hierarchical multistage {G}aussian signaling games in noncooperative communication and control systems.
\newblock {\em Automatica}, 107:9--20, 2019.

\bibitem{velicheti2023value}
R.~K. Velicheti, M.~Bastopcu, and T.~Ba{\c{s}}ar.
\newblock Value of information in games with multiple strategic information providers.
\newblock {\em arXiv preprint arXiv:2306.14886}, 2023.

\bibitem{saritacs2016quadratic}
S.~Sar{\i}ta{\c{s}}, S.~Y{\"u}ksel, and S.~Gezici.
\newblock Quadratic multi-dimensional signaling games and affine equilibria.
\newblock {\em IEEE Transactions on Automatic Control}, 62(2):605--619, 2016.

\bibitem{castiglioni2020online}
M.~Castiglioni, A.~Celli, A.~Marchesi, and N.~Gatti.
\newblock Online {B}ayesian persuasion.
\newblock {\em Advances in Neural Information Processing Systems}, 33:16188--16198, 2020.

\bibitem{hazan2012near}
E.~Hazan, S.~Kale, and S.~Shalev-Shwartz.
\newblock Near-optimal algorithms for online matrix prediction.
\newblock In {\em Conference on Learning Theory}, pages 38--1. JMLR Workshop and Conference Proceedings, 2012.

\bibitem{balcan2015commitment}
M.~F. Balcan, A.~Blum, N.~Haghtalab, and A.~D. Procaccia.
\newblock Commitment without regrets: Online learning in {S}tackelberg security games.
\newblock In {\em Proceedings of the Sixteenth ACM Conference on Economics and Computation}, pages 61--78, 2015.

\bibitem{9222530}
M.~O. Sayin and T.~Başar.
\newblock Persuasion-based robust sensor design against attackers with unknown control objectives.
\newblock {\em IEEE Transactions on Automatic Control}, 66(10):4589--4603, 2021.

\bibitem{cover1991information}
T.~M. Cover and J.~A. Thomas.
\newblock Information theory and statistics.
\newblock {\em Elements of Information Theory}, 1(1):279--335, 1991.

\bibitem{christiano2014online}
P.~Christiano.
\newblock Online local learning via semidefinite programming.
\newblock In {\em Proceedings of the Forty-sixth Annual ACM Symposium on Theory of Computing}, pages 468--474, 2014.

\bibitem{moridomi2018online}
K.~Moridomi, K.~Hatano, and E.~Takimoto.
\newblock Online linear optimization with the log-determinant regularizer.
\newblock {\em IEICE Transactions on Information and Systems}, 101(6):1511--1520, 2018.

\bibitem{zinkevich2003online}
M.~Zinkevich.
\newblock Online convex programming and generalized infinitesimal gradient ascent.
\newblock In {\em Proceedings of the 20th International Conference on Machine Learning (ICML-03)}, pages 928--936, 2003.

\bibitem{hazan2007logarithmic}
E.~Hazan, A.~Agarwal, and S.~Kale.
\newblock Logarithmic regret algorithms for online convex optimization.
\newblock {\em Machine Learning}, 69(2):169--192, 2007.

\bibitem{hazan2012projection}
E.~Hazan and S.~Kale.
\newblock Projection-free online learning.
\newblock {\em arXiv preprint arXiv:1206.4657}, 2012.

\bibitem{hazan2016introduction}
E.~Hazan et~al.
\newblock Introduction to online convex optimization.
\newblock {\em Foundations and Trends{\textregistered} in Optimization}, 2(3-4):157--325, 2016.

\bibitem{kalai2005efficient}
A.~Kalai and S.~Vempala.
\newblock Efficient algorithms for online decision problems.
\newblock {\em Journal of Computer and System Sciences}, 71(3):291--307, 2005.

\bibitem{caratheodory1911variabilitatsbereich}
C.~Carath{\'e}odory.
\newblock {\"U}ber den variabilit{\"a}tsbereich der fourier’schen konstanten von positiven harmonischen funktionen.
\newblock {\em Rendiconti Del Circolo Matematico di Palermo (1884-1940)}, 32(1):193--217, 1911.

\end{thebibliography}

\section*{Appendix A} 
{\bf Proof of Lemma \ref{lem: etasig}:}
    The first part of the lemma can be proved by considering the positive semi-definite matrix $\mathbb{E}[(\omega-\hat{\omega})(\omega-\hat{\omega})^\top]\succeq O$ and noting that $\mathbb{E}[\omega\hat{\omega}^\top]=O$ by using the principle of orthogonality. To prove the second part, first notice that any positive semi-definite matrix with $\Sigma_0\succeq\Sigma\succeq0$ can be written as $\Sigma=U_0^\top\Lambda_0^\frac{1}{2}T\Lambda_0^\frac{1}{2}U_0$, where $I\succeq T\succeq0$ and $\Sigma_0=U_0^\top\Lambda_0U_0$. Let $T=U_t^\top\Lambda_tU_t$. Now consider a linear policy with noise, i.e., $\eta(\omega)=L\mathbf{\omega}+\mathbf{w}$ for some matrix L of appropriate dimensions and let $\mathbf{w}\in\mathbb{R}^d$ be an independent zero-mean Gaussian noise with covariance $\Theta\succeq O$, which is to be determined. This implies $\mathbb{E}[\hat{\omega}\hat{\omega}^\top]=\Sigma_0L(L^\top\Sigma_0L+\Theta)^\dagger L^\top\Sigma_0$. Now consider $L=U_0\Lambda^{-\frac{1}{2}}_0U_t\Lambda_l$. Using $\Lambda_l=diag(\{\lambda_{l,1},\hdots,\lambda_{l,d}\})$ and $\Theta=diag(\{\theta_1^2,\hdots,\theta_d^2\})$ such that $\frac{\lambda_{l,i}^2}{\lambda_{l,i}^2+\theta_i^2}=\lambda_{t,i}$, where $\Lambda_t=diag(\{\lambda_{t,1}\hdots,\lambda_{t,d}\})$, results in $\mathbb{E}[\hat{\omega}\hat{\omega}^\top]=\Sigma$. Since $\lambda_{t,i}<1$, this satisfies $I\succeq T\succeq O$. 
\section*{Appendix B}
\begin{lemma}\label{lem: projmat}
   The boundaries of the convex set $\mathscr{W}=\{W|I\succeq W\succeq O\}$ are projection matrices. 
\end{lemma}
\vspace{-0.5 em}
\begin{proof}
    This can be proved by contradiction. Let $W_P\in\mathscr{W}$ be a projection matrix and hence has eigenvalues $(\lambda_P)_i\in\{0,1\}$, where $(\lambda_P)_i$ is the $i$th eigenvalue. Suppose that there exist matrices $P,Q\in\mathscr{W}$ and $P,Q\notin\mathscr{M}$, where $\mathscr{M}=\{W|I\succeq W\succeq O,(\lambda_W)_i\in\{0,1\}\forall i\}$, such that $tP+(1-t)Q=W_P$. Let $p_n,p_r$ be orthonormal eigenvectors of $W_p$ such that $p_n^\top W_Pp_n=0$, $p_r\top W_Pp_r=1$. This implies $tp_n^\top Pp_n+(1-t)p_n^\top Q p_n=0$ and $tp_r^\top Pp_r+(1-t)p_r^\top Q p_r=1$. But since $P,Q\in\mathscr{W}$, we have $0\leq p^\top Pp\leq1$ and $0\leq p^\top Q p\leq 1$ for any $p$. Therefore, $p_n,p_r$ are eigenvectors of $P,Q$ with eigenvalues $0,1$, respectively. This shows that $P$ and $Q$ are also projection matrices.
\end{proof}

\end{document}